\documentclass{article}

\usepackage[english]{babel}

\usepackage[letterpaper,top=2cm,bottom=2cm,left=3cm,right=3cm,marginparwidth=1.75cm]{geometry}

\usepackage{amsmath}
\usepackage{graphicx}
\usepackage{amssymb}
\usepackage{amsthm}
\usepackage{subcaption}
\newtheorem{theorem}{Theorem}
\newtheorem{proposition}{Proposition}
\usepackage[boxed]{algorithm2e}
\usepackage{algpseudocode}
\usepackage{enumerate}
\usepackage[colorlinks=true, allcolors=blue]{hyperref}
\newcommand{\parallelsum}{\mathbin{\!/\mkern-5mu/\!}}

\title{Evaluation of acoustic Green's function in rectangular rooms with general surface impedance walls}
\author{Matteo Calaf\`a\footnote{\url{matcal@dtu.dk}}, Yuanxin Xia, Jonas Brunskog, Cheol-Ho Jeong\\\\
Acoustic Technology, Department of Electrical and Photonics Engineering, \\Technical University of Denmark, 2800 Kongens Lyngby, Denmark}
\date{}
\begin{document}
\maketitle

\begin{abstract}
Acoustic room modes admit closed-form expressions for rectangular rooms with perfectly reflecting walls, from which the Green's function can be computed directly through the eigenfunction expansion. First-order approximations also exist for nearly rigid boundaries; however, current analytical methods fail to accommodate more general boundary conditions, e.g., when wall absorption is significant. In this work, we present a comprehensive analysis that extends previous studies by including additional first-order asymptotics that account for soft-wall boundaries. In addition, we introduce a semi-analytical, efficient, and reliable method for computing the Green's function in rectangular rooms, which is described and verified through numerical tests. The resulting error decreases rapidly for sufficiently large truncation order, making the method suitable as a benchmark for numerical simulations. Additional aspects regarding the spectral basis orthogonality and completeness are also addressed, providing a general framework for the validity of the proposed approach.
\end{abstract}

\section{\label{sec:introduction} Introduction}
The room impulse response (RIR) plays a fundamental role in room acoustics, as it allows the derivation of several key acoustic parameters, including reverberation time and clarity \cite{ISO3382-1:2009}, and it generally incorporates comprehensive information of the enclosed space \cite{vorlander2008auralization,kuttruff2016room}. Moreover, under the assumption of a linear time-invariant (LTI) system, convolving the RIR with an arbitrary input signal directly yields the corresponding output sound field. The RIR can be obtained from the Green's function, defined as the solution to the Helmholtz equation under a Dirac delta source \cite{duffy2015green,jacobsen2013fundamentals,okoyenta2020short}. Consequently, the computation of Green's functions is of paramount practical importance in many acoustic applications.

Besides established numerical techniques such as the finite element method (FEM) \cite{harari2006survey} and the boundary element method (BEM) \cite{preuss2022recent}, Green's functions can also be obtained through spectral representations based on expansions in terms of an orthogonal basis of room eigenfunctions \cite{sztranyovszky2025extending, habermann2025asymptotic}. However, this approach has several limitations: the eigenfunction set is infinite and typically unavailable in closed form. As a result, the eigenfunction expansion (EE) must be truncated and the single terms are usually computed numerically, leading to substantial computational cost, particularly when a large number of eigenfunctions is required.

In certain cases, explicit formulas for the eigenfunctions are available and can be exploited to construct the Green's function efficiently. For example, \cite{perez2010green} derived the series representation for the boundary value problem in the cylindrical domain. The most notable example, however, is the rectangular room with perfectly reflecting boundaries, for which closed-form modal solutions are well documented in classical acoustics literature \cite{morse_theoretical_1968, jacobsen2013fundamentals, kuttruff2016room}. Based on this, approximations can be obtained for walls that are not perfectly rigid. Specifically, \cite{morse_sound_1944} is among the earliest studies of first-order asymptotic expansions of eigenfunctions in rectangular rooms with surface admittance boundary conditions. More recently, \cite{nolan_two_2019} revisited underlying aspects of the derivation, including the orthogonality of the basis. Nevertheless, although these studies account for non-rigid boundaries, they remain limited to sufficiently hard walls, as they rely on asympotics that assume the admittance to lie in a neighborhood of the origin.

Formulations applicable to more general surface admittance values have also been proposed, although they typically rely on numerical schemes that reduce the interpretability of the solutions and require iterative procedures. For example, \cite{bistafa2003numerical} introduced an iterative approach in which the eigenvalue solutions are obtained for progressively increasing impedance values. In contrast, \cite{naka2005acoustic} employed an interval Newton/generalized bisection (IN/GB) method, capable of guaranteeing the detection of all solutions within prescribed bounds. Finally, \cite{du2011acoustic} developed an algorithm based on the Rayleigh–Ritz method, where the eigenfunctions are represented as Fourier series expansions.

This work aims to derive the Green's function through its eigenfunction expansion. The proposed formulation extends the Morse solution beyond hard walls, while still yielding interpretable first-order asymptotic approximations of the solutions. As a result, the theory provides the basis for an efficient and reliable algorithm for evaluating the Green’s function.

The manuscript is organized as follows: Section~\ref{sec:physicalmodel} introduces the physical models and the related mathematical problems under investigation. Section~\ref{sec:analysis} contains the main contribution of this work, namely the derivation of first-order eigenvalue approximations and the resulting method to evaluate the Green's function. Later, Section~\ref{sec:validityof} examines the assumptions of basis orthogonality and completeness required for the application of the eigenfunction expansion. An auxiliary eigenvalue problem is studied in Section~\ref{sec:Auxiliary eigenvalue problem}, which provides additional insight on the characterization of the Green's function in rectangular rooms. Finally, some numerical and experimental tests are performed in Section~\ref{sec:numericaltests} and conclusions are reported in Section~\ref{sec:conclusions}.

\section{\label{sec:physicalmodel} Problem formulation}
\subsection{Eigenfunction expansion for the Green's function}
We consider a connected and bounded domain $\Omega\subset \mathbb{R}^d$ in $d\in\mathbb{N}$ spatial dimensions. Under the time-harmonic condition, the acoustic field can be described by the Helmholtz equation with given boundary conditions (BCs). Specifically, imposing the normalized surface admittance $\beta:\partial\Omega\rightarrow\mathbb{C}$ (or, equivalently, the normalized surface impedance $\zeta=1/\beta$) on the boundaries provides a general and accurate modeling framework for a wide range of physical configurations. The resulting boundary value problem (BVP) consists of finding the acoustic pressure $p:\Omega\rightarrow\mathbb{C}$ such that

\begin{equation}\label{eq:bvp}
    \begin{cases}
    \nabla^2 p + k^2 p = s, & \text{ in } \Omega, \\
    \nabla p \cdot \mathbf{n} + i k \beta p = 0, & \text{ on } \partial \Omega,
    \end{cases}
\end{equation}
where $\nabla^2(\cdot)$ is the Laplacian operator, $s:\Omega\rightarrow\mathbb{C}$ is the source term and $k\in\mathbb{R}$ is the excitation wavenumber. In particular, the Green's function is defined as the solution $G_k(\cdot|\mathbf{x}_0):\Omega\rightarrow\mathbb{C}$ to Equation \eqref{eq:bvp} when the source is a Dirac delta function $s(\mathbf{x})=-\delta(\mathbf{x}-{\mathbf{x}_0})$ at $\mathbf{x}_0 \in \Omega$. For simplicity, dissipation in the medium is neglected, although this effect can be incorporated by allowing $k$ to have a nonzero imaginary part. The particular case $\beta=0$ corresponds to a Neumann BC, modeling perfectly rigid walls that reflect incident waves without dissipation or phase delay.

We introduce the associated \emph{eigenvalue problem} as follows: given $k\in\mathbb{R}$, find $\hat{k}\in\mathbb{C}$ and $\varphi\neq 0$ such that
\begin{equation}\label{eq:eigenvalueproblem}
    \begin{cases}
    \nabla^2 \varphi + \hat{k}^2 \varphi = 0, & \text{ in } \Omega, \\
        \nabla \varphi \cdot \mathbf{n} + i k \beta \varphi = 0, & \text{ on } \partial \Omega.
    \end{cases}
\end{equation}
Namely, the eigenpair $(\hat{k},\varphi)$ is identified by the eigenvalue $\lambda:=-\hat{k}^2$ of the Laplacian operator and the eigenfunction $\varphi$. Eigenpairs are introduced as they provide a means to represent the Green's function. Specifically, under some conditions discussed more in detail in Section \ref{sec:validityof}, the following eigenfunction expansion (EE) holds:
\begin{equation}\label{eq:greenfunctionseries}
    G_{k}(\mathbf{x}|\mathbf{x}_0) = \sum_{n=1}^\infty\frac{\varphi_n(\mathbf{x})\varphi_n(\mathbf{x}_0)}{\Lambda_n(\hat{k}_n^2-k^2)},
\end{equation}
where $\{(\hat{k}_n,\varphi_n)\}_{n=1}^\infty$ denote all distinct solutions to Equation (\ref{eq:eigenvalueproblem}) and $\Lambda_n:=\int_\Omega \varphi^2_n \, d\mathbf{x}$ is a normalization constant.

Note that $\beta$, $s$, and the eigenpairs generally depend on $k$. This dependence is omitted for notational brevity, also because $k$ is assumed fixed throughout the following sections unless stated otherwise. Finally, while $k$ is real-valued, $\hat{k}$ may have a nonzero imaginary part, physically representing the acoustic energy dissipation introduced by $\beta \neq 0$.

\subsection{Problem formulation in rectangular rooms}\label{sec:problemformulationinrectangularrooms}

We restrict the geometry $\Omega$ to the rectangular room and assume $\beta$ to be spatially uniform on each wall side. This implies that $\beta$ can be identified by a vector of dimension $2d$. Under these assumptions, the technique of separation of variables (SoV) can be applied, and the multi-dimensional problem in Equation~\eqref{eq:eigenvalueproblem} can be reduced to a 1D problem ($d=1$) for each single axis. In particular, one utilizes
\begin{gather*}
    \varphi_n(\mathbf{x}) = \prod_{j=1}^d \varphi_{n_j}(\mathbf{x}_j), \hspace{5mm}
    \hat{k}^2_
    n= \sum_{j=1}^d \hat{k}_{n_j}^2,
\end{gather*}
where $d$ is the space dimension, $n$ is now a multi-index $n=(n_1,n_2,\dots,n_d)\in \mathbb{N}^d$, and each term on the right hand side is obtained from the 1D problem. 
The described SoV strategy follows the approach adopted in several previous studies, to which the interested reader is referred for further details \cite{morse_theoretical_1968,bistafa2003numerical,nolan_two_2019}. Accordingly, the remainder of this section focuses exclusively on the 1D problem.

Let us define $l\in\mathbb{R}^+$ the room size along the axis under consideration and $\beta_-,\beta_+\in\mathbb{C}$ the normalized surface admittance on the left and right wall, respectively. The general solution to the 1D homogeneous Helmholtz equation can be written as \begin{equation}\label{eq:generalsolution1D}
    \varphi(x) = \cos\left(\frac{\pi}{l}\hat{q}x + \hat{b}\right), \quad x\in \left[-\frac{l}{2},\frac{l}{2}\right],
\end{equation}
where $\hat{q},\hat{b}\in\mathbb{C}$ need to be found by imposing the boundary conditions. In particular, $\hat{q}=\hat{k}l/\pi$ is proportional to $\hat{k}$ and the normalization constant in Equation \eqref{eq:greenfunctionseries} can be obtained explicitly as
\begin{equation}\label{eq:lambda}
    \Lambda = \frac{l}{2}\left(1 + \frac{\sin\left(\pi \hat{q}\right)\cos\left(2\hat{b}\right)}{\pi \hat{q}}\right).
\end{equation}
Imposing Equation~\eqref{eq:generalsolution1D} into Equation~\eqref{eq:eigenvalueproblem}, leads to the condition on $\hat{q}$:
\begin{equation}\label{eq:main_condition}
    \left( (\pi \hat{q})^2 + \gamma_-\gamma_+\right) \tan \left(\pi \hat{q}\right) = i (\gamma_-+\gamma_+)(\pi \hat{q}),
\end{equation}
where $\gamma_\pm:=\beta_\pm kl$. $\hat{b}$ can instead be found by solving 
\begin{equation}\label{eq:b}
    \pi \hat{q} \tan\left(\frac{\pi}{2}\hat{q} + \hat{b}\right) = i\gamma_+,
\end{equation} given $\hat{q}$. Let us denote with $v(\hat{q})$ the residual of Equation~\eqref{eq:main_condition}; $v(\hat{q})=0$ represents the main problem to address in the following sections, as eigenpairs can be directly inferred from $\hat{q}$ and SoV can be employed to extend the solution to 2D or 3D rectangular rooms. Specifically, the equation is transcendental and admits no closed-form solution except in a few special cases. Among these, we highlight the choice of $\beta_-=\beta_+=0$, which reduces to the problem with perfectly rigid walls and whose solutions are $\hat{q}\in\mathbb{Z}$, as established by classical theory. Note that solutions are real-valued only in the perfect rigid scenario, whereas $\hat{q}$ is generally complex otherwise.

We emphasize that Equation \eqref{eq:main_condition} can be found in the literature in several related forms. See, for instance, Equation (5.4) in \cite{morse_sound_1944},  Equation (12) in \cite{bistafa2003numerical}, Equation (5) in \cite{naka2005acoustic} and Equation (8) in \cite{nolan_two_2019}.

\section{\label{sec:analysis} Derivation of the method}
In this section, we formulate the method to evaluate the Green's function from its eigenfunction expansion. Section \ref{sec:Properties of eigenvalue solutions} and \ref{sec:firstorderasymptotics} are devoted to the analysis and approximation of the solutions to Equation~\eqref{eq:main_condition}, respectively. Such approximated solutions are then used in Section \ref{sec:Evaluation of eigenfunction expansion} for evaluating the Green's function.

\subsection{Properties of 1D eigenvalue solutions}\label{sec:Properties of eigenvalue solutions}
Some properties of Equation~\eqref{eq:main_condition} can be readily deduced: $\hat{q}=0$ is a solution and $+\hat{q}$ is a solution if and only if $-\hat{q}$ is. This is expected from the fact that $\lambda=-\hat{k}^2$ is the eigenvalue of the problem. For such reason, we define $\mathcal{Q}\equiv \mathcal{Q}_-\cup\mathcal{Q}_+\cup\{0\}$ the total set of solutions, by discerning the solution copies on the left and right side of the complex plane, respectively. In particular, we focus on only one side, say $\mathcal{Q}_+$. Furthermore, let $\#_m \mathcal{Q_+}$ denote the number of solutions $\hat{q}$ in $\mathcal{Q}_+$ such that $|\hat{q}|\le m+\frac{1}{2}$.

 We proceed by proving the following.

\begin{theorem}[Number of solutions]\label{theo:1}
There exists $m_0\in\mathbb{N}$ such that for any $m\ge m_0$: \begin{equation*}
    \#_{m}\mathcal{Q}_+ = \begin{cases}
        m, & \text{if } \gamma_-\gamma_+=i(\gamma_-+\gamma_+), \\
         m + 1, & \text{otherwise}.
    \end{cases}
\end{equation*}
\end{theorem}
\begin{proof}
Define $v_1(\hat{q})=((\pi \hat{q})^2+\gamma_-\gamma_+)\sin(\pi \hat{q})$ and $v_2(\hat{q}) =-i(\gamma_-+\gamma_+)(\pi \hat{q}) \cos(\pi \hat{q})$, so that Equation~\eqref{eq:main_condition} can be written as $v(\hat{q})=v_1(\hat{q}) + v_2(\hat{q}) = 0$. Then, we have
\begin{equation*}
    \sup_{\hat{q}\in \partial B_{m}} \frac{|v_2(\hat{q})|}{|v_1(\hat{q})|} \le \frac{|\gamma_-+\gamma_+|\pi(m+\frac{1}{2})}{\pi^2(m+\frac{1}{2})^2 - |\gamma_-\gamma_+|} \sup_{\hat{q}\in \partial B_{m}} |\cot(\pi \hat{q})|,
\end{equation*}
where $\partial B_m$ denotes the circumference of radius $m+\frac{1}{2}$.
The first factor tends to zero as $m\rightarrow \infty$, while the second factor can be bounded by 2 for any $m$. Thus, by Rouché´s theorem (see, e.g., Theorem 10.43 in \cite{rudinRealComplexAnalysis1986}), the cardinality of $\mathcal{Q}$ restricted to $B_m$ is equal to that of the zeros of $v_1(\hat{q})$, i.e., $\mathbb{Z}\cap B_m$ and $\pm \sqrt{-\gamma_+\gamma_-}/\pi$. Finally, if $\gamma_-\gamma_+=i(\gamma_-+\gamma_+)$, the solution $\hat{q}=0$ has triple multiplicity, single otherwise. By subtraction, the cardinality of $\mathcal{Q}_+$ can be determined.
\end{proof}

In our setting, all solutions to Equation~\eqref{eq:eigenvalueproblem} must be found. Theorem~\ref{theo:1} therefore provides a useful criterion for assessing whether the complete set of eigenvalues has been obtained. It is worth noting that, unlike the perfectly rigid case, the general problem admits an additional nontrivial eigenvalue.

We next establish another useful property, showing that the imaginary part of the solutions remains bounded except in a neighborhood of $\gamma_\pm/\pi$.

\begin{proposition}[Bound on imaginary component]\label{prop:imag}
Let $\hat{q}$ be a solution to Equation~\eqref{eq:main_condition}. Then, for any $L>0$, either $|\text{Im}(\pi\hat{q})| \le L$ or $|\pi \hat{q} - \gamma| \le 2|\gamma| / (e^L-1)$ for some $\gamma\in\{\gamma_-,\gamma_+\}$.
\end{proposition}
\begin{proof}
Equation~\eqref{eq:main_condition} can be written as
    \begin{equation*}
        \text{Im}(\pi\hat{q}) = \frac{1}{2}\ln\left(\frac{|\pi\hat{q}+\gamma_-||\pi\hat{q}+\gamma_+|}{|\pi\hat{q}-\gamma_-||\pi\hat{q}-\gamma_+|}\right).
    \end{equation*}
Therefore, one has $|\text{Im}(\pi\hat{q})| \le |\ln\left(|\pi\hat{q}+\gamma| / |\pi\hat{q}-\gamma|\right)|$ for either $\gamma_-$ or $\gamma_+$. Assume by contradiction that $|\text{Im}(\pi\hat{q})|>L$, then 
\begin{equation*}
    |\pi\hat{q}-\gamma| < \frac{|\pi\hat{q}+\gamma|}{e^L} \le \frac{|\pi\hat{q}-\gamma| + 2|\gamma|}{e^L},
\end{equation*}
where the triangular inequality is applied in the last step. By rearranging the terms, we complete the proof.
\end{proof}

\subsection{First-order asymptotics of 1D eigenvalues}\label{sec:firstorderasymptotics}
Approximations of the solutions to Equation~\eqref{eq:main_condition} can be found through the use of asymptotics of the $\tan(\cdot)$ function. Namely:

\begin{enumerate}[i]
    \item Taylor expansion around zeros: $\text{tan}(\pi\hat{q}) \approx \epsilon$ for $\hat{q}=n+\epsilon$, $n\in\mathbb{N}_0$, $|\epsilon|\ll 1$.
    \item Laurent expansion around poles: $\tan(\pi \hat{q}) \approx -\frac{1}{\pi \epsilon}$ for $\hat{q}=n+\frac{1}{2}+\epsilon$, $n\in\mathbb{N}_0$, $|\epsilon|\ll 1$ .
    \item Limit for imaginary infinity: $\tan(\pi \hat{q}) \approx (i-2ie^{2i\pi \hat{q}})$ for $\text{Im}(\hat{q})\gg 0$.
\end{enumerate}

\begin{table}[ht]
\centering
\begin{tabular}{c|c c c c c}
    \hline \hline
    Group & Interpretation & Assumption &\hspace{4mm} Index $n$ \hspace{4mm}  &  Derivation & Approximated $\hat{q}$\\ 
    \hline 
     1 & Hard walls & $|\gamma_-|,|\gamma_+|\ll n$  & $0,1,2,\dots$ & i, $\Xi$ &  $ \frac{1}{2}\left(n + \sqrt{n^2+4i\frac{\gamma_-+\gamma_+}{\pi^2}}\right)$ \\
     2 & Soft walls & $|\gamma_-|,|\gamma_+|\gg n$  & $1,2,\dots $ & i, $\neg \Xi$ & $n \left(1+\frac{i}{(\gamma_-\parallelsum \gamma_+)-i}\right)$ \\ 
     1P & Asymmetric walls & $ |\gamma_\pm|\gg n,
            |\gamma_\mp| \ll n$ & $0,1,2,\dots $ & ii, $\Xi$ & $\left(n+\frac{1}{2}\right)\left(1+\frac{i}{\gamma_-+\gamma_+}\right)$\\
     \hline
    \rule[-2ex]{0pt}{5.5ex} 3 & Positive susceptance  & $\text{ Im}(\gamma_\pm) \gg 0$ & / &  iii & $\dfrac{\gamma_\pm}{\pi}$ \\
    \hline \hline
\end{tabular}
\caption{Overview of the asymptotic groups for the 1D eigenvalue problem in Equation~\eqref{eq:main_condition}.}
\label{tab:groups}
\end{table}
\vskip1in

In addition, we introduce the condition $\Xi: |\hat{q}|^2 \gtrsim |\gamma_-\gamma_+|$ and its negation $\neg \Xi:  |\hat{q}|^2 \lesssim |\gamma_-\gamma_+|$. Combining each asymptotic regime of $\tan(\cdot)$ with either $\Xi$ or $\neg \Xi$ yields a group of approximate solutions to Equation~\eqref{eq:main_condition}, each corresponding to a specific range of validity. For brevity, the derivations are omitted and the resulting groups are summarized in Table~\ref{tab:groups}. Here, $\gamma_- \parallelsum \gamma_+ := \frac{\gamma_-\gamma_+}{\gamma_-+\gamma_+}$,
which corresponds to the equivalent resistance of two resistors connected in parallel.

Group 1, 2 and 1P are complementary to each other, so that any combination of $\beta_-,\beta_+$ belongs to one of these groups or at their intersection. Specifically, group 1 corresponds to the regime of sufficiently hard walls, group 2 to sufficiently absorbing walls, and group 1P to highly asymmetric boundary conditions, where one wall is hard and the other is absorbing. Group 3 is instead distinct as it provides one or two more solutions based on a fundamentally different assumption. In particular, group 3 directly corresponds to the result stated in Proposition~\ref{prop:imag}.

A practical criterion for selecting the correct group has been established analytically and confirmed numerically: first, the solution(s) from group 3 are valid provided that $\text{Im}(\gamma_{\pm})\gtrsim 1$. Then, one examines whether walls are highly asymmetric by checking that the ratio $|\gamma_\pm|/|\gamma_\mp|\gtrsim 5$. If the condition holds, then for each $n$, the formula from group 2 is employed if $n\lesssim A_{2-1P} := |\gamma_-\parallelsum\gamma_+|/2$, group 1P is instead adopted between $A_{2-1\text{P}}$ and $A_{1\text{P}-1}:=|\gamma_-+\gamma_+|/\pi$, while group 1 is used for $n\gtrsim A_{1_\text{P}-1}$. If walls are instead not highly symmetric, then group 1P should be disregarded; in this case, only groups 1 and 2 are adopted, with the transition occurring at $A_{2-1}:=\sqrt{|\gamma_-\gamma_+|}/\pi$.

For perfectly symmetric walls ($\gamma=\gamma_-=\gamma_+$) the two solutions in group 3 collapse to the same value. In this special case, a corrected expression can be derived. Let $\hat{q}=(\gamma+\epsilon)/\pi$ for $|\epsilon|\ll1$, where $\epsilon=0$ yields the original solution. By applying Equation~\eqref{eq:main_condition} one can readily obtain $\epsilon^2 \approx 4 \gamma^2 e^{2i(\gamma+\epsilon)}$. This equation can be solved in terms of the Lambert $W$ function \cite{mezo2022lambert}. Retaining only the leading-order term gives the corrected approximation $\hat{q}\approx \frac{\gamma}{\pi}(1\pm 2e^{i\gamma})$.

Additional groups can be derived formally; however, they are based on mutually inconsistent assumptions and therefore do not yield accurate approximations in practice. This is the case for the hypothetical group 2P (obtained from ii and $\neg\Xi$) as well as for the counterpart of group 3 corresponding to the regime $\operatorname{Im}(\gamma_\pm)\ll 0$.

\subsection{Evaluation of eigenfunction expansion}\label{sec:Evaluation of eigenfunction expansion}

Algorithm~\ref{alg:1} outlines the step-by-step procedure for evaluating the Green's function in 1D at a point $\mathbf{x}$ due to an impulse source located at $\mathbf{x}_0$, based on the results established in the preceding sections. Namely, first-order asymptotic approximations from Section \ref{sec:firstorderasymptotics} are used as effective initial choices for a root-finding algorithm that refines the solutions of Equation~\eqref{eq:main_condition}. The Newton-Raphson method is chosen due to its simplicity, as $v(\cdot)$ and its first derivative $v'(\cdot)$ can be obtained in closed form. Nevertheless, other root-finding algorithms may be adopted. The truncated EE in Equation~\eqref{eq:greenfunctionseries} is then applied, where $n_{\max}$ must at least be greater than $q := k l/\pi$, ensuring that the dominant terms in the series are included. The criterion from Proposition \ref{prop:imag} can possibly be included to discard solutions that have not converged.

The proposed method can be compared with earlier approaches in the literature; specifically, the well-known Morse first-order approximation \cite{morse_theoretical_1968} employs only the asymptotics from group 1, which, as discussed in Section~\ref{sec:firstorderasymptotics}, is accurate only for sufficiently hard walls. This limitation can be mitigated by applying the root-finding algorithm, as done in \cite{nolan_two_2019}, although highly inaccurate initial guesses may obstruct convergence. In contrast, by incorporating all asymptotic groups, suitable initial guesses are obtained for all solutions.

Although Algorithm~\ref{alg:1} is formulated for the one-dimensional case only, the SoV strategy introduced in Section~\ref{sec:problemformulationinrectangularrooms} readily extends the computation to multiple dimensions by applying the same construction independently in each coordinate direction. In particular, the total number of iterations becomes $(n_{\text{max}} + 1)^d$, where $d$ denotes the spatial dimension.

\begin{algorithm}[!ht]
\caption{Evaluation of Green's function (1D case).}
\label{alg:1}

\KwIn{$\beta_-,\beta_+\in\mathbb{C}$}
\KwIn{$l,k\in\mathbb{R}^+$}
\KwIn{$\mathbf{x},\mathbf{x}_0\in\Omega$}
\KwIn{$n_\text{max},n_\text{newton}\in\mathbb{N}$}
\KwIn{$\alpha_\text{newton},\varepsilon_\text{newton}\in\mathbb{R}^+$}

$\mathcal{Q}_+ \gets \emptyset$\;

\For{$n=0,\ldots,n_\text{max}$}{
    \ForEach{Group $\in \{1,2,3,1\mathrm{P}\}$}{
        \If{Assumption from Group (Table \ref{tab:groups})}{
            Calculate $\hat q_n$ from Table \ref{tab:groups}
            and add to $\mathcal{Q}_+$\;
        }
    }
}

\For{$n=0,\ldots,n_\text{newton}$}{
    $\mathcal{Q}_+
    \gets
    \mathcal{Q}_+
    -\alpha_\text{newton}
    \dfrac{v(\mathcal{Q}_+)}
          {v'(\mathcal{Q}_+)}$\;

    $\varepsilon
    \gets
    \|v(\mathcal{Q}_+)\|_1$\;
}

\If{$\varepsilon>\varepsilon_\text{newton}$}{
    Raise warning\;
}

Remove copies from $\mathcal{Q}_+$\;

Remove $\hat q=0$ from $\mathcal{Q}_+$\;

\If{Cardinality of $\mathcal{Q}_+ \neq n_\text{max}+1$
(Theorem~\ref{theo:1})}{
    Raise error\;
}

$G_k \gets 0$\;

\For{$n=0,\ldots,n_\text{max}$}{
    $\hat q_n \gets$ $n$-th element of $\mathcal{Q}_+$\;

    Get $\hat b_n$ from Equation~\eqref{eq:b}\;

    Calculate $\varphi_n$ at $\mathbf{x},\mathbf{x}_0$
    from Equation~\eqref{eq:generalsolution1D}\;

    Get $\Lambda_n$ from Equation~\eqref{eq:lambda}\;

    Calculate $G_{k,n}$ from
    Equation~\eqref{eq:greenfunctionseries}\;

    $G_k \gets G_k + G_{k,n}$\;
}

\end{algorithm}

\section{\label{sec:validityof}Validity of eigenfunction expansion}

The EE in Equation \eqref{eq:greenfunctionseries} lies on the assumption that eigenfunctions form an orthogonal and complete basis. While the spectral theorem guarantees these properties when $\beta=0$, it does not apply when $\beta\neq 0$ because the underlying operators are no longer self-adjoint. In this section, we therefore rigorously verify that the EE can still be used for non-rigid walls.

\subsection{Orthogonality of eigenfunctions}
We recall that eigenfunctions are real-valued only for $\beta=0$, therefore the analysis is carried out in $L^2(\Omega)\equiv L^2(\Omega;\mathbb{C})$ instead. Let 
$\langle \varphi_n,\varphi_m\rangle_{L^2(\Omega)} := \int_\Omega\varphi_n\overline{\varphi_m} \, d\mathbf{x}$ denote the standard Hermitian inner product, where $\overline{\varphi}$ is the complex conjugate of $\varphi$. It turns out that eigenfunctions are generally not Hermitian, i.e., $\langle \varphi_n,\varphi_m\rangle_{L^2(\Omega)}\neq 0$, however they satisfy
\begin{equation}\label{eq:pseudoorthogonality}
\langle \varphi_n,\overline{\varphi_m}\rangle_{L^2(\Omega)} =  \int_\Omega \varphi_n \varphi_m \, d\mathbf{x}= 0, \quad \text{for } n\neq m.
\end{equation}
This property was used in \cite{morse_theoretical_1968} without however an explicit demonstration. A proof was later provided for rectangular domains in \cite{nolan_two_2019} and subsequently extended to general domains in \cite{badeau_spectral_2025}. In what follows, we show that the property is in fact a consequence of a substantially more general mathematical result, from which a concise proof follows naturally.

\begin{theorem}\label{theo:orthogonality}
Let $T$ be a function space and $\mathcal{B} : T \times T \rightarrow \mathbb{C}$ a bilinear form. Suppose $\mathcal{L} : T \to T$ is a $\mathcal{B}$-self-adjoint operator, i.e., $\mathcal{B}(\mathcal{L}u_1, u_2) = \mathcal{B}(u_1, \mathcal{L}u_2)$, $\forall u_1, u_2 \in T$. 
Then, any two eigenfunctions $\varphi_1,\varphi_2$ of $\mathcal{L}$ corresponding to distinct eigenvalues are $\mathcal{B}$-orthogonal, i.e., $\mathcal{B}(\varphi_1, \varphi_2) = 0$.
\end{theorem}
\begin{proof}
    The proof is immediate and follows exactly the same passages adopted for the standard orthogonality. Namely, let $\lambda_n\neq\lambda_m$ be the eigenvalues associated to $\varphi_n,\varphi_m$, respectively. Then,
    \begin{equation*}
    \begin{aligned}
        \lambda_n\mathcal{B}(\varphi_n,\varphi_m) &= \mathcal{B}(\mathcal{L}\varphi_n, \varphi_m) \\ 
        &= \mathcal{B} (\varphi_n, \mathcal{L}\varphi_m) = \lambda_m \mathcal{B} (\varphi_n, \varphi_m),
    \end{aligned}
    \end{equation*}
    which implies that $ \mathcal{B} (\varphi_n, \varphi_m)=0$.
\end{proof}

Theorem~\ref{theo:orthogonality} is stated in a broad setting and applies whenever self-adjointness is defined with respect to a bilinear form, rather than a sesquilinear inner product. It applies in our context by choosing $\mathcal{B}(\cdot,\cdot)=\langle \cdot,\overline{\cdot}\rangle_{L^2(\Omega)}$. The resulting class of $\mathcal{B}$-self-adjoint operators is large, and in particular includes any operator of the form $\mathcal{L}=\sum_{j}\eta_j\mathcal{L}_j$, with $\eta_j\in\mathbb{C}$ and $\mathcal{L}_j$ Hermitian. It follows that the Helmholtz problem in Equation \eqref{eq:bvp} satisfies Equation~\eqref{eq:pseudoorthogonality} for any complex-valued parameters $k$ and $\beta$, i.e., even when surface admittance or dissipation through the medium are taken into account.

An alternative way to interpret this result is to observe that, for such a class of operators, the eigenfunctions of $\mathcal{L}$ and those of its adjoint $\mathcal{L}^*$ are complex conjugates of each other, i.e., $\overline{\varphi_n}=\varphi^*_n$, $\forall n\in\mathbb{N}$. Consequently, Equation~\eqref{eq:pseudoorthogonality} follows automatically by the biorthogonality property \cite{zhedanov1999biorthogonal}.

In conclusion, Equation~\eqref{eq:greenfunctionseries} requires no modification for $\beta\neq 0$ since it already accounts for the $\mathcal{B}$-orthogonality.

\subsection{Completeness of the spectral basis}
Completeness of the basis was long assumed to be true, and has only recently been examined in detail in \cite{badeau_spectral_2025}. The corresponding analysis is technical, and it requires the distinction between different definitions of completeness. We do not address these details here and instead refer the interested reader to the aforementioned work.

The central difficulty arises from the fact that the underlying operator is not self-adjoint, as discussed above, which obstructs the application of the spectral theorem. In general, a complete system may consist of both eigenfunctions and \emph{generalized} eigenfunctions \cite{keldysh_completeness_1971}, which are defined iteratively as
\begin{equation*}
    \mathcal{L}\varphi_{n,j} = \lambda_n \varphi_{n,j} + \varphi_{n,j-1}, \quad  j\in\mathbb{N},
\end{equation*}
where $\varphi_{n,0}=\varphi_n$ is the standard $n$-th eigenfunction. However, it can be seen that $\varphi_n$ admits associated generalized eigenfunctions only if
\begin{equation*}
\begin{aligned}
    \mathcal{B}(\varphi_n, \varphi_n) &= \mathcal{B}(\mathcal{L}\varphi_{n,1},\varphi_n) -  \mathcal{B}(\lambda_n\varphi_{n,1},\varphi_n) \\
    &= \mathcal{B}(\varphi_{n,1},\mathcal{L}\varphi_n) -  \mathcal{B}(\varphi_{n,1},\lambda_n\varphi_n) =0,
\end{aligned}
\end{equation*}
where we have used the assumption of Theorem~\ref{theo:orthogonality} on the second line. Notably, this condition is equivalent to $\Lambda_n=0$, implying a zero denominator in Equation~\eqref{eq:greenfunctionseries}.
We can then state the following.

\begin{theorem}\label{theo:completeness}
    The solutions $\varphi$ to the problem in Equation \eqref{eq:eigenvalueproblem} in the $d$-dimensional rectangular room with admittances $\beta_1,\beta_2,\dots,\beta_{2d}$ on each of the $2d$ walls form a complete basis for all admittance values except for a set $\Theta\subset \mathbb{C}^{2d}$ of zero Lebesgue measure.
\end{theorem}

\begin{proof}
    The set $\Theta$ in 1D is identified by the solutions of
    \begin{equation*}
        \begin{cases}
            \sin(\pi \hat{q})\cos\left(\hat{b}\right) = -\pi \hat{q}, \\
            \frac{\pi}{2}\hat{q} \tan\left(\frac{\pi}{2}\hat{q}-\hat{b}\right) = i\gamma_-, \\
            \frac{\pi}{2}\hat{q} \tan\left(\frac{\pi}{2}\hat{q}+\hat{b}\right) = i\gamma_+,
        \end{cases}
    \end{equation*}
    where the first equation follows by applying $\Lambda = \langle \varphi,\overline{\varphi}\rangle_{L^2(\Omega)}=0$ from Equation~\eqref{eq:lambda}, and the remaining two equations arise as intermediate steps in the derivation of \eqref{eq:main_condition}. \\
    Take any $\hat{q}\in\mathbb{C}$, then the associated $\hat{b}$ is unique from the first equation except for multiples of $\pi$. It follows from the other two equations that the associated $\gamma_-,\gamma_+$ are unique. Furthermore, the map is holomorphic. Consequently, the set of admissible $\beta_-,\beta_+$ is contained in the image of a differentiable function from $\mathbb{C}$ to $\mathbb{C}^2$, which has zero measure \cite{tao2011introduction}. The extension to $d>1$ readily follows.
\end{proof}

To summarize, we have verified that the orthogonality holds in a significantly larger setting and completeness holds for almost every admittance parameter. Therefore, Equation~\eqref{eq:greenfunctionseries} is guaranteed to hold as long as the problem parameters do not belong to the corresponding zero measure set.

\section{Auxiliary eigenvalue problem}\label{sec:Auxiliary eigenvalue problem}

In this section, we introduce and briefly study a variation of the original eigenvalue problem in Equation~\eqref{eq:eigenvalueproblem}. Specifically: find $\tilde{k}\in\mathbb{C}$ and $\phi\neq 0$ such that
\begin{equation}\label{auxiliaryproblem}
    \begin{cases}
    \nabla^2 \phi + \tilde{k}^2 \phi = 0, & \text{ in } \Omega, \\
        \nabla \phi \cdot \mathbf{n} + i \tilde{k} \beta  \phi = 0, & \text{ on } \partial \Omega.
    \end{cases}
\end{equation}
This problem is equivalent to Equation~\eqref{eq:eigenvalueproblem} except the two terms $k$ and $\hat{k}$ are now represented by the same unknown $\tilde{k}$. However, note that the auxiliary eigenvalue $\tilde{k}$ is in general different from both $k,\hat{k}$, since $\tilde{k}$ is a solution to an apparently similar but distinct problem. Furthermore, some of the properties of the original eigenfunctions, such as the $\mathcal{B}$-orthogonality and the SoV, do not hold in this context. Nevertheless, it is shown that the study of such problem can be still useful for the analysis of the Green's function in rectangular rooms.

 By assuming $\beta$ in Equation~\eqref{auxiliaryproblem} to be frequency-independent, both the problem and the solution $\tilde{k}$ are independent of the excitation wavenumber $k$. Consequently, unlike in the previous sections where $k$ was fixed, the analysis can be extended directly to any excitation frequency. The solution to the 1D auxiliary eigenvalue problem can be readily obtained by replacing both $k,\hat{k}$ in Equation~\eqref{eq:main_condition} with $\tilde{k}$, which yields:
\begin{equation}\label{eq:auxiliary_rectangular}
    \tilde{q} = \frac{1}{\pi} \arctan\left(i\frac{\beta_-+\beta_+}{1+\beta_-\beta_+}\right) + n, \quad n \in\mathbb{Z},
\end{equation}

where we denote $\tilde{q}=\tilde{k}l/\pi$ analogously to $\hat{q}$. Different from the original eigenvalue problem, the solution here is exact and can be obtained explicitly. Note that $\arctan(\cdot)$ has to be considered as the main branch of the associated multivalued function.
Finally, observe that there is no solution symmetry along the real axis, in contrast to the eigenvalues examined earlier.

It is convenient to analyze the real and imaginary components separately. From Equation~\eqref{eq:auxiliary_rectangular}, we obtain
\begin{equation}\label{eq:auxiliary_real}
     \text{Re}(\tilde{q}) = \dfrac{1}{2\pi}\arg(R_-R_+) +n, \quad  n\in\mathbb{Z},
\end{equation}
\begin{equation}\label{eq:auxiliary_imag}
    \text{Im}(\tilde{q}) = - \dfrac{1}{2\pi}\ln|R_-R_+|,
\end{equation}
where $R=(1-\beta)/(1+\beta)$ is the pressure reflection coefficient at the normal incidence angle. If $\tilde{q}$ is real-valued, it follows directly that the system admits an excitation wavenumber $k$ equal to one of its eigenvalues $\hat{k}$, resulting in a singular Green’s function. Consequently, $\text{Im}(\tilde{q})$ provides a meaningful measure of damping, which from Equation~\eqref{eq:auxiliary_imag} is null only for $|R_- R_+| = 1$. Furthermore, the damping is maximized when either $R_-$ or $R_+$ is equal to zero, corresponding as expected to perfectly absorbing 1D boundaries.

The real part $\text{Re}(\tilde{q})$ in Equation~\eqref{eq:auxiliary_real} is, on the other hand, particularly useful for identifying the locations of resonance frequencies in the transfer function. Indeed, the excitation wavenumber $k$ is closest to $\tilde{k}$ when evaluated at $\text{Re}(\tilde{k})$. When $k$ and $\tilde{k}$ are similar, then, by definition of $\tilde{k}$, they are expected to lie near an eigenvalue $\hat{k}$, thereby producing resonance. Finally, note that each value of $\text{Re}(\tilde{q})$ in Equation~\eqref{eq:auxiliary_real} is correctly confined to the interval $[n, n+1)$.

\section{\label{sec:numericaltests}Results}
In this section, we verify and validate the proposed methods using both numerical simulations and experimental measurements.. For all numerical experiments, we consider the speed of sound in air $c=343$~m/s. Codes have been implemented in \texttt{Python 3.12} with minimal dependencies, in particular relying on \texttt{NumPy 2.2.3} for vector computations \cite{harris2020array}. The implementation of Algorithm~\ref{alg:1} will be made available after the publication of this manuscript\footnote{https://github.com/dtu-act/green-function-rect-rooms.git}.

\subsection{Approximation of eigenvalues}
We begin by verifying the first-order approximations derived in Section~\ref{sec:firstorderasymptotics}. Therefore, we consider the 1D problem with length $l=1$~m, and $f= c k / (2\pi) = 5000$~Hz. Note that the choice of these parameters is largely arbitrary, since in Equation~\eqref{eq:main_condition} they affect only $\gamma_\pm = \beta_\pm k l$, which also depends on $\beta_\pm$. Accordingly, we fix these parameters and evaluate the quality of the approximations by varying $\beta_\pm$ instead.

In order to assess the accuracy of the first-order asymptotics, we consider the quantity $v_L:=\log_{10}(|v|+1)$, where $v(\hat{q})$ was introduced in Section \ref{sec:problemformulationinrectangularrooms}. The purpose of $v_L$ is to numerically estimate the solutions $\hat{q}$ by visualizing the contours on a more appropriate logarithmic scale. Here, $v_l(\hat{q})$ is evaluated from $\text{Re}(\hat{q})\in [0,8], \text{Im}(\hat{q})\in[0,5]$ on a uniform grid $500\times 800$. Three choices of admittance $(\beta_-,\beta_+)$ are selected, and contour plots are reported in Figure~\ref{fig:test1}. Exact solutions are identified by the red areas where $v_L=0$, while approximations are displayed with different markers according to each group. 

\begin{figure}[!ht]
\centering
\begin{subfigure}[h]{0.5\textwidth}
\includegraphics[width=\textwidth]{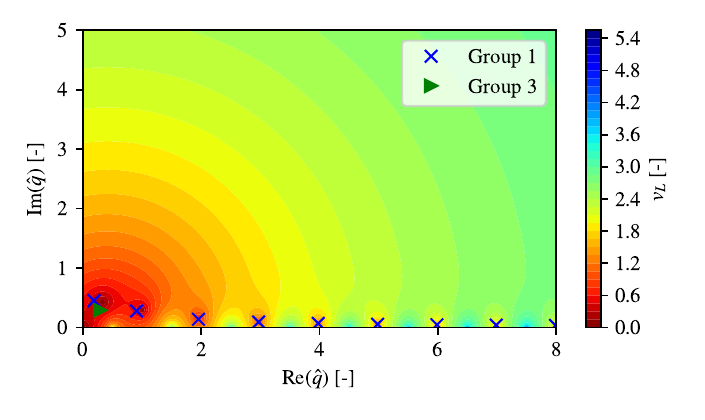}
\caption{(A) $\beta_-=0.01+0.01i$, $\;\; \beta_+=0.02$, $\;\; A_{2-1}\approx 0.49$.}
\end{subfigure}
\begin{subfigure}[h]{0.5\textwidth}
\includegraphics[width=\textwidth]{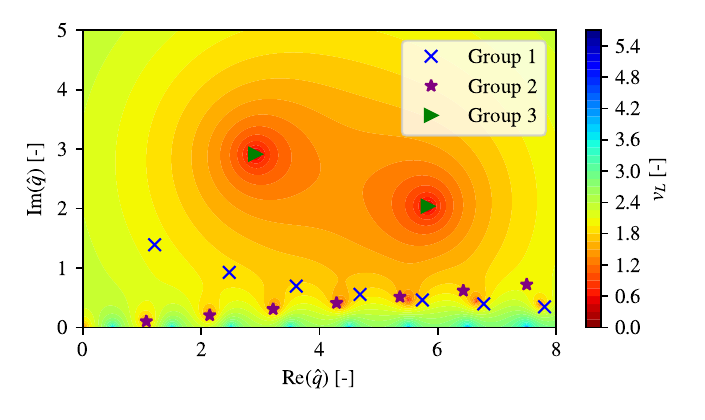}
\caption{(B) $\beta_-=0.1+0.1i$, $\;\; \beta_+=0.2+0.07i$, $\;\; A_{2-1}\approx 5.05$.}
\end{subfigure}
\begin{subfigure}[h]{0.5\textwidth}
\includegraphics[width=\textwidth]{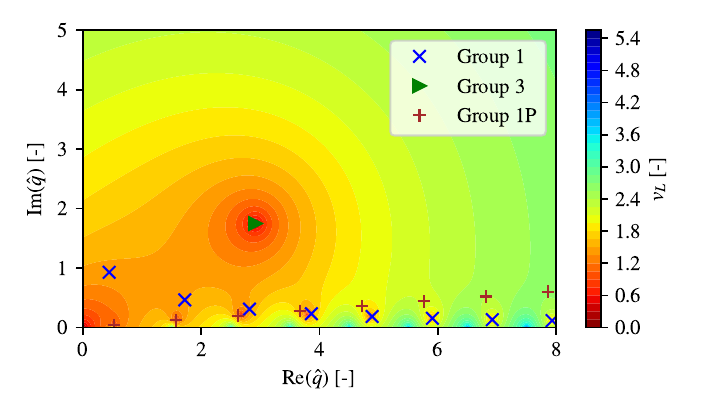}
\caption{(C) $\beta_-=0.1+0.06i$, $\; \beta_+=0$, $\; A_{2-1\text{P}}=0$, $\; A_{1\text{P}-1}\approx 3.40$.}
\end{subfigure}
\caption{Evaluation of eigenvalue first-order asymptotics from Section~\ref{sec:firstorderasymptotics} for different admittance values. (A) Hard walls. (B) Soft walls. (C) Highly asymmetric walls. Admittance and cutoff transition values are provided.}
\label{fig:test1}
\end{figure}

Figure~\ref{fig:test1}A displays a scenario characterizing hard walls. As expected, group 1 is sufficient to well represent all solutions. Group 3 is also added as the susceptance on the left wall is positive and equal to the conductance. However, it can be noted that this is redundant, as it overrides group 1 to approximate the solution closer to the origin. We recall that $\hat{q}=0$ is also present as the trivial solution. As expected, $A_{2-1}<1$, indicating that groups 2 should not be included. 

More absorbing walls are considered in Figure~\ref{fig:test1}B, where $A_{2-1}\approx 5.05$ indicates that group 2 should be included for $n\lesssim 5$ in place of group 1, as confirmed by the graphics. Moreover, group 3 is included twice, as both susceptances are not negligible and positive. Group 1P is instead omitted since $\beta_-$ and $\beta_+$ are of the same order.

A final scenario is shown in Figure~\ref{fig:test1}C, where one wall is weakly soft with positive susceptance, while perfect reflection is applied on the opposite side. In this case, group 1P should be included up to $n< A_{1-1\text{P}} \approx 3.40$, and $A_{2-1\text{P}} = 0$ indicates that group 2 should instead be disregarded, as confirmed by the plot.

Figure~\ref{fig:test1} additionally validates Theorem~\ref{theo:1}, as 9=8+1 solutions are visible in each of the three contour plots.

We emphasize again that Figure~\ref{fig:test1}A corresponds to scenarios with hard walls examined in previous studies \cite{morse_theoretical_1968, nolan_two_2019}, where group 1 was sufficient to approximate all eigenvalues. By contrast, increasing either the admittance order or the excitation frequency necessitates the inclusion of additional asymptotic groups, as illustrated in Figures~\ref{fig:test1}B and~\ref{fig:test1}C.

\subsection{Calculation of Green's function}
The eigenfunction expansion (EE) from Algorithm~\ref{alg:1} is here employed to calculate the Green's function. A 2D rectangular room of size $1.0$ $\times$ $1.4$~m is tested with excitation frequency $f=5000$ Hz, and the source point $\mathbf{x}_0$ is located at the coordinates [0.2~m, 0.2~m] with respect to the center of the room. Furthermore, we assign a different normalized surface admittance on each of the four walls: $\beta_{-,x}=0.09+0.03i, \beta_{+,x}=0.16, \beta_{-,y}=0.07+0.03i, \beta_{-,y}=0.12+0.12i$. The EE is employed with $n_\text{max}=160 \gg q$, as recommended in Section~\ref{sec:Evaluation of eigenfunction expansion}. In addition, $n_\text{newton}=100$, $\alpha_\text{newton}=0.3$ and the tolerance to identify solution copies or the trivial $\hat{q}=0$ is set to $\texttt{tol}=10^{-4}$.

We first qualitatively benchmark our algorithm against a reference solution obtained from a conventional high-resolution numerical simulation. Specifically, a FEM solver is implemented in \texttt{FreeFEM 4.15} \cite{hecht2012} using second-order Lagrange (P2) nodal elements. A uniform mesh with approximately 15 elements per wavelength (EPW) is employed, which is known to provide high accuracy \cite{Marburg2008}. This discretization results in $n_{DOF}=380{,}557$ degrees of freedom. To treat the Dirac delta source, a singularity-removal strategy is adopted (see for instance \cite{gjerde2020singularity, holtershinken2025local}). Namely, the Green's function is decomposed as $G_{k}(\cdot \mid \mathbf{x}_0) = \Psi_{k}(\cdot \mid \mathbf{x}_0) + G^0_{k}(\cdot \mid \mathbf{x}_0)$, where
\begin{equation*}
\Psi_{k}(\mathbf{x}|\mathbf{x}_0) = \begin{cases}
    \frac{i}{4}H^{(1)}_0(k|\mathbf{x}-\mathbf{x}_0|), & d=2, \\
    \frac{1}{4\pi|\mathbf{x}-\mathbf{x}_0|} e^{ik |\mathbf{x}-\mathbf{x}_0|}, & d=3,
\end{cases}
\end{equation*}
is the free-space fundamental solution of the Helmholtz equation, which captures the singularity at $\mathbf{x}_0$, and $H_0^{(1)}(\cdot)$ denotes the Hankel function of the first kind. Therefore, the FEM problem is addressed to find only the singularity-free component $G^0_{k}$.

The comparison of the two computed solutions is displayed in Figure~\ref{fig:test2}, where they appear nearly identical, with only a minimal discrepancy at the singular source point. Interestingly, the effect of SoV in the eigenfunction expansion is visible with an increased error on the horizontal and vertical axes in correspondence of $\mathbf{x}_0$.

\begin{figure}[!ht]
\centering
\begin{subfigure}[h]{0.65\textwidth}
\includegraphics[width=\textwidth]{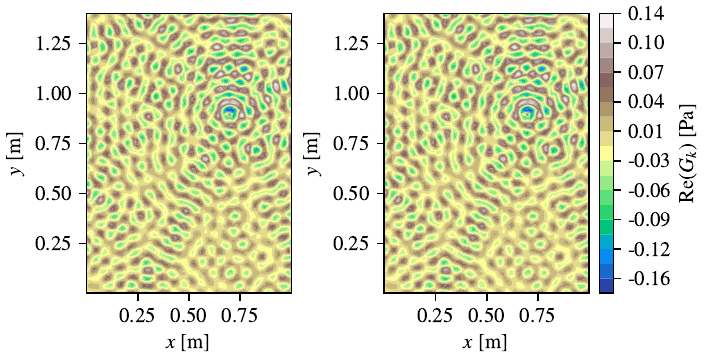}
\caption{(A) $\text{Re}(G_{k})$ (EE). \hspace{12mm} (B) $\text{Re}(G_{k})$ (FEM).}
\end{subfigure}
\begin{subfigure}[h]{0.65\textwidth}
\includegraphics[width=\textwidth]{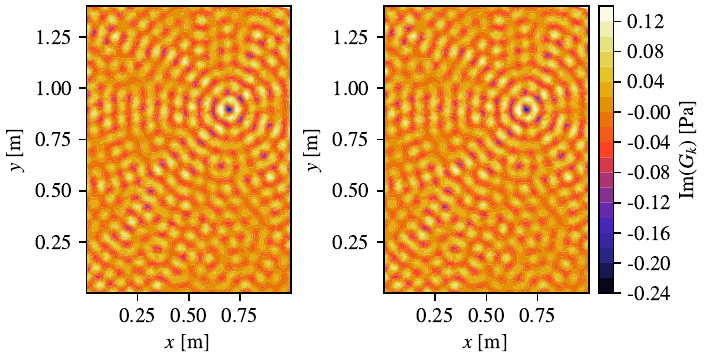}
\caption{(C) $\text{Im}(G_{k})$ (EE). \hspace{12mm} (D) $\text{Im}(G_{k})$ (FEM).}
\end{subfigure}
\begin{subfigure}[h]{0.44\textwidth}
\includegraphics[width=\textwidth]{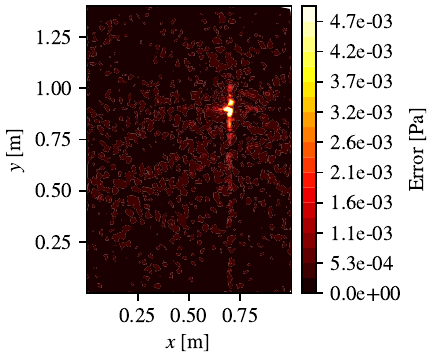}
\caption{(E) Absolute value error.}
\end{subfigure}
\caption{Comparison of the EE (Algorithm~\ref{alg:1}) and FEM simulation. The same color map is applied on each row. The two solutions appear considerably similar with a slightly larger error at the source singular location.}
\label{fig:test2}
\end{figure}

To verify that the EE can serve as a benchmark for numerical simulations, a FEM grid refinement study is carried out by progressively increasing the number of EPW from 3 to 15. We then consider the $L^2(\Omega)$ relative error:
\begin{equation*}
    \mathcal{E}(p,p_{ref}) := \sqrt{\frac{\int_\Omega |p-p_{ref}|^2 \, d\mathbf{x}}{\int_\Omega |p_{ref}|^2 \, d\mathbf{x}}},
\end{equation*}
where the integrals are approximated by evaluating the terms on a large number ($10^4$) of uniformly sampled points in $\Omega$, excluding those within a distance smaller than 0.05 from $\mathbf{x}_0$. The convergence to the EE solution is shown in Figure \ref{fig:test2convergence}A, where a third-order rate is observed, in agreement with the classical theory given by the Aubin-Nitsche lemma \cite{ciarlet2002finite}). The error relative to the Morse solution is also shown, highlighting its inadequacy for this problem. This outcome is expected, since $A_{2-1} \approx 5$ along both axes, indicating that the inclusion of group 2 is necessary.

\begin{figure}[!ht]
    \centering
    \begin{subfigure}[h]{0.5\textwidth}
        \includegraphics[width=\textwidth]{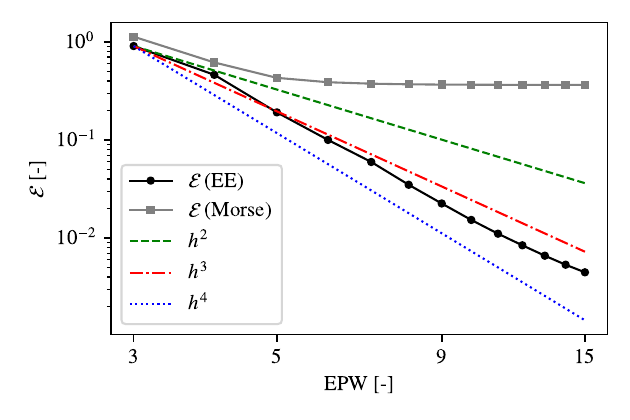}
        \caption{(A)}
    \end{subfigure}
    \begin{subfigure}[h]{0.5\textwidth}
        \includegraphics[width=\textwidth]{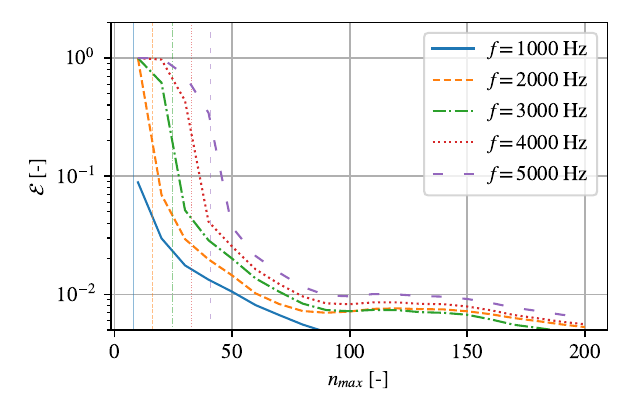}
        \caption{(B)}
    \end{subfigure}
\caption{Convergence studies: (A) FEM grid refinement converging to the EE solution at $f=5000$ Hz, $n_\text{max}=200$. Third-order is observed. (B) Convergence of the EE for increasing truncation order $n_\text{max}$ towards the FEM solution at EPW=15. Vertical lines are placed in correspondence of $q$, which as expected coincide with the faster convergence rate.}
\label{fig:test2convergence}
\end{figure}

We additionally aim to assess the influence of the truncation order $n_\text{max}$ and the problem frequency on the overall accuracy. Errors are then computed between the evaluations from the EE and the high-fidelity FEM solution (EPW=15) for different values of $n_{\max}$ and $f$. The results are shown in Figure~\ref{fig:test2convergence}B, indicating robust performance across a range of frequencies, with slightly faster convergence for lower ones. It is also readily observed that each curve exhibits a pronounced change in slope around $n_{\max} = q$, as expected from the discussion in Section~\ref{sec:Evaluation of eigenfunction expansion}.

A fair comparison of computational cost is challenging because of the fundamentally different methodologies and implementations (compiled versus interpreted language). Nevertheless, the total number of floating point operations in Algorithm~\ref{alg:1} is approximately $ 10(n_{\max}+1)^d \approx 4 \times 10^5 $, with $d = 2$ and $n_{\max} =200$. This is of the same order as $n_{\mathrm{DOF}}$ for EPW=15. We therefore conclude that, in this respect, the EE approach is significantly more efficient than a traditional FEM implementation, both in terms of memory usage and computational time. This advantage is expected to become even more pronounced as room size or frequency increases.

\subsection{Analysis of auxiliary eigenvalues}
We dedicate a brief section to verify the results from Section~\ref{sec:Auxiliary eigenvalue problem} on auxiliary eigenvalues. Specifically, we consider a 1D domain of length $l=1$~m, with a source at $\mathbf{x}_0=0.2$~m and a receiver at $\mathbf{x}=0.8$~m. The EE is used on a frequency range 500~Hz $\le f\le$ 3000~Hz with a resolution of 5~Hz. The sound pressure level (SPL) is then evaluated as SPL = $20 \log_{10}(|G_k(\mathbf{x})|/p_0)$, where $G_k(\mathbf{x})$ is the computed Green's function from Algorithm~\ref{alg:1} at the point $\mathbf{x}$ and $p_0=2\cdot 10^{-5}$ Pa is the reference sound pressure in air.

\begin{figure}[!ht]
\centering
\begin{subfigure}[h]{0.5\textwidth}
\includegraphics[width=\textwidth]{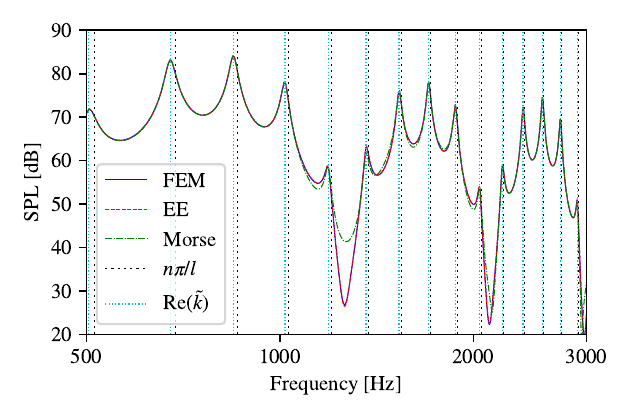}
\caption{(A) $\beta_-=\beta_+=0.1+0.1i$.}
\end{subfigure}
\begin{subfigure}[h]{0.5\textwidth}
\includegraphics[width=\textwidth]{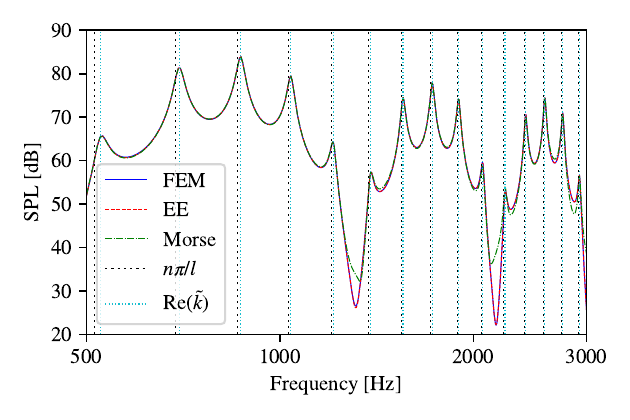}
\caption{(B) $\beta_-=\beta_+=0.1-0.1i$.}
\end{subfigure}
\begin{subfigure}[h]{0.5\textwidth}
\includegraphics[width=\textwidth]{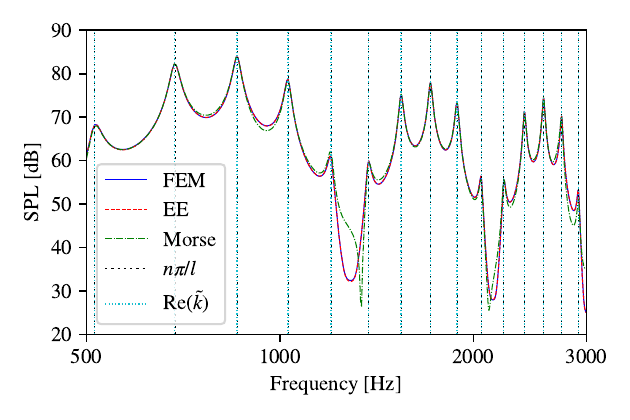}
\caption{(C) $\beta_-=0.1-0.1i, \quad \beta_+=0.1+0.1i$.}
\end{subfigure}
\caption{Comparison of Green's function and analytical auxiliary eigenvalues from Equation~\eqref{eq:auxiliary_real}.}
\label{fig:test3}
\end{figure}

The SPL is plotted in Figure~\ref{fig:test3} for three different choices of $\beta_-,\beta_+$, along with the $\text{Re}(\tilde{q})$ values from Equation~\eqref{eq:auxiliary_real}. As expected, the vertical lines align with the peaks of the transfer function in all cases. To further examine how the peak locations depend on the boundary admittance, the eigenvalues for perfectly rigid boundaries (given by the classical formula $n\pi/l$) are also included. Figure~\ref{fig:test3}A exhibits a leftward shift, in contrast to the rightward shift observed in Figure~\ref{fig:test3}B, while Figure~\ref{fig:test3}C shows an exact overlap. These outcomes are expected from Equation~\eqref{eq:auxiliary_real} since the condition $\beta_1=\beta_2$ leads to $\arg(R_1 R_2)>0$ if and only if $\text{Im}(\beta_1)=\operatorname{Im}(\beta_2)<0$ and the case $\beta_1=\overline{\beta_2}$ yields $\arg(R_1 R_2)=0$.

A high-fidelity FEM solution with EPW=15 and the Morse first-order approximation are also included. The EE solution and the FEM results show a perfect match, whereas the MM solution exhibits noticeable discrepancies again. This behavior is expected, as $A_{2-1}$ reaches values up to 4. Nevertheless, the MM solution remains accurate in the vicinity of the resonance peaks. This can be explained by observing that, near a peak, a single eigenvalue $\hat{k}$ lies very close to $k$ and therefore dominates the remaining contributions in the EE series. The corresponding index is $n \approx kl/\pi$, so that the validity condition for group 1, namely $n \gtrsim A_{2-1}$, reduces to $|\beta_-\beta_+| \lesssim 1$, which is satisfied in the present case.

It should be emphasized that the values selected for $\beta_-,\beta_+$ are primarily intended for illustrative purposes and for verifying the theory developed in Section~\ref{sec:Auxiliary eigenvalue problem}. A detailed investigation of materials exhibiting such admittance values is beyond the scope of this work.

\subsection{Comparison with measurements}
In this last experiment, we examine the method's ability to predict the impulse response in a real room by comparing it with measurements. 

The room under investigation is an approximately rectangular space displayed in Figure~\ref{fig:room}. Measurements were carried out using a National Instruments 4431 (National Instruments, Austin, TX) data acquisition system, with signal amplification provided by a HBK 2734 amplifier (HBK, Nærum, Denmark). Acoustic responses were recorded using a HBK 4191 microphone, connected through type 2690 conditioning amplifiers and a HBK NEXUS preamplifier. Excitation was provided by an exponential sine sweep sampled at 96 kHz from an HBK Type 4292-L omnidirectional loudspeaker, and the resulting transfer functions were computed via deconvolution with the inverse sweep followed by fast Fourier transform processing. 
\begin{figure}[ht]
    \centering
    \includegraphics[width=0.49\linewidth]{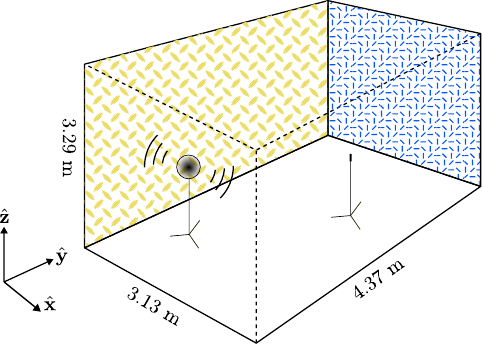}\hfill
    \includegraphics[width=0.45\linewidth]{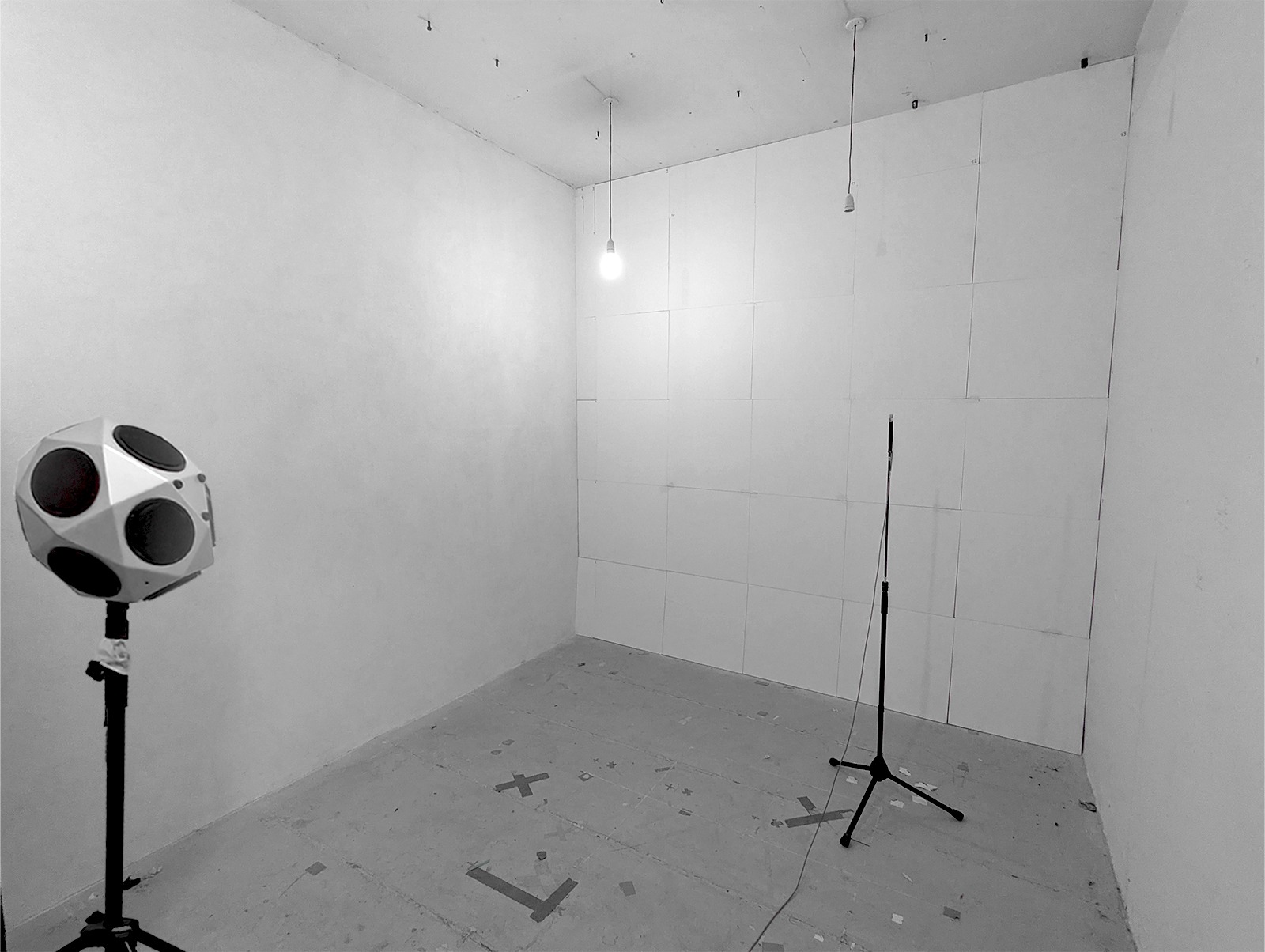}
    \caption{Sketch and photo of the acoustic measurement room. Ecophon Master A (blue) is placed at the wall in the $y$ direction opposite to the origin, Ecophon Industry Modus (yellow) is later added on the face corresponding to $x=0$.}
    \label{fig:room}
\end{figure}

The empty room walls consist of concrete, whose absorption coefficient was estimated using Eyring's formula \cite{eyring1930reverberation}, yielding the frequency-independent value of $\alpha_{\text{concrete}} = 0.015$. The admittance, $\beta \approx 4 \times 10^{-3}$, was directly derived from $\alpha_{\text{concrete}}$ under the assumption of negligible susceptance.

Two configurations are considered. In Config 1, Ecophon Master A panels (Ecophon, Hyllinge, Sweden) are installed on one of the six room walls without air gap (average $\alpha_{\text{Master}}\approx 0.60$). In Config 2, Ecophon Industry Modus panels (average $\alpha_{\text{Modus}}\approx 0.83$) are additionally mounted on a second wall, as illustrated in Figure~\ref{fig:room}. Config 1 is mildly absorptive overall, yielding a reverberation time of 2.25~s. Adding the second panel reduces the reverberation time to 0.72~s. We note that installing the panels modify the dimensions of the room, being 4 and 10 cm thick, respectively. The acoustic surface admittance of both panel types was measured \emph{in situ} using the device described in \cite{xia2026surface}.

Algorithm~\ref{alg:1} was applied using the same parameter settings as in the previous test. The resulting predictions are compared with the measurements in Figure~\ref{fig:test4} over the frequency range 150~Hz $\leq f \leq$ 600~Hz, with a frequency resolution of 0.1~Hz. The Morse approximated solution is included in the plots, as well as the Schroeder transition frequency \cite{schroeder1962frequency}.

\begin{figure}[!ht]
\centering
\begin{subfigure}[h]{0.5\textwidth}
\includegraphics[width=\textwidth]{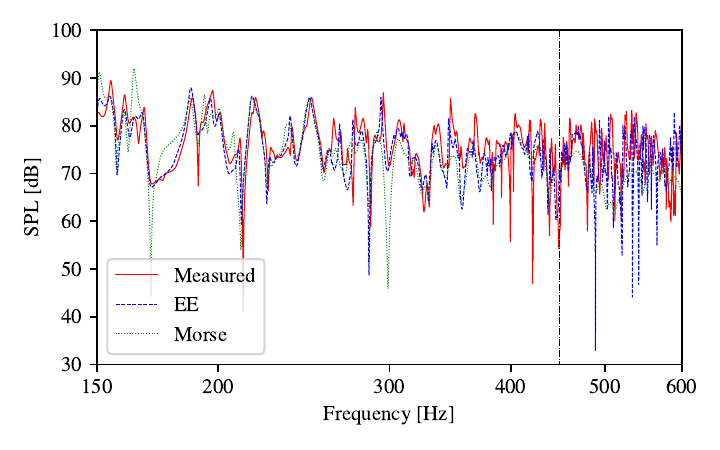}
\caption{(A) Config 1.}
\end{subfigure}
\begin{subfigure}[h]{0.5\textwidth}
\includegraphics[width=\textwidth]{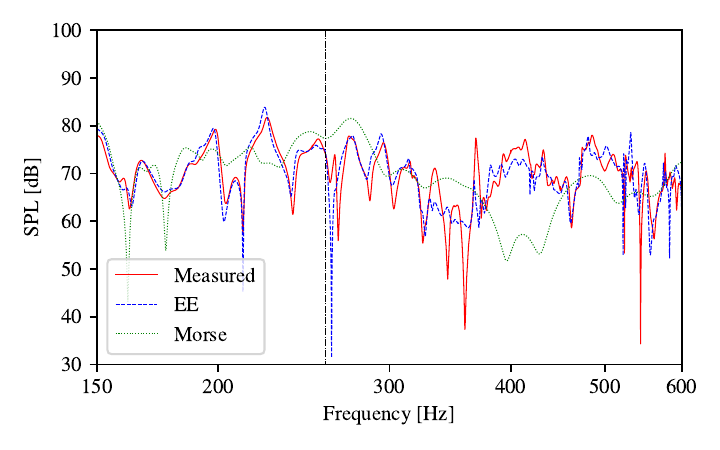}
\caption{(B) Config 2.}
\end{subfigure}
\caption{Comparison of measured and computed SPL in the rectangular room. The Schroeder frequency is marked as vertical line. (A) $\mathbf{x}=[2.00, 3.31,1.67]$ m, $\mathbf{x}_0=[1.32,1.10,1.89]$ m. (B) $\mathbf{x}=[1.49, 3.63,1.68]$ m, $\mathbf{x}_0=[1.87,1.73,1.88]$ m.}
\label{fig:test4}
\end{figure}

The proposed EE method shows good agreement with the measurements in both configurations, particularly below the Schroeder frequency. In contrast, the Morse solution exhibits substantially lower accuracy, especially in Config 2, where the higher level of absorption violates the assumptions of the approximation. 

To quantitatively assess the agreement between the solutions, we employ the frequency response assurance criterion (FRAC), a widely used similarity metric for comparing transfer functions \cite{heylen1996frac}. The FRAC ranges from 0 to 1, where a value of 1 indicates perfect shape similarity and a value of 0 corresponds to completely dissimilar responses. The resulting FRAC values are reported in Table~\ref{tab:frac}. Overall, the EE yields higher scores than the Morse model, particularly in Config 2. This behavior is expected, as the increased absorption leads to a smoother and more regular response, thereby improving the agreement between the solutions. Conversely, the Morse model performs slightly better in Config 1, as its underlying hard-wall assumption is more consistent with the room conditions.

\begin{table}[ht]
    \centering
    \begin{tabular}{c|c|c}
    \hline \hline
    Configuration & Method & FRAC (\%) \\
    \hline 
    Config 1 & EE & 54.1 \\
     & Morse & 41.7 \\
    \hline
    Config 2 & EE & 78.1 \\
     & Morse & 26.5 \\
    \hline \hline
    \end{tabular}
    \caption{Evaluation of FRAC with respect to measurements.}
    \label{tab:frac}
\end{table}

The analyzed frequency range could readily be extended, given the negligible computational cost of executing Algorithm~\ref{alg:1}. On the other hand, the estimated surface impedance of the panels is not sufficiently reliable below 150 Hz due to the limitation of the speaker, while analysis at higher frequencies is hindered by the increased density of peaks in the transfer function and measurement uncertainties arising from geometrical imperfections and other experimental limitations. Similar limitations have been reported in the literature; for example, \cite{luan2008method} observed good agreement with measurements only below 320~Hz.

Regarding the use of asymptotic groups from Section~\ref{sec:firstorderasymptotics},
the two configurations employ both group 1 and 1P due to the high axial asymmetry. In particular, $A_{1-1\text{P}}$ reaches up to the value of 7 at 600 Hz in the $y$ axis with the Master A panel, and the value of 8 in the $x$ axis when the Industry Modus panel is applied.

\section{\label{sec:conclusions}Conclusions}
This work investigates the use of eigenfunction expansions to construct the Green's function in rectangular rooms with surface admittance boundary conditions. Since closed-form solutions cannot be obtained, first-order asymptotic approximations are developed. Notably, four distinct families of asymptotics are identified, each corresponding to a different admittance regime. This allows the calculation of the Green's function for a wide range of admittance values, in contrast to previous studies that focused on a single family associated with sufficiently rigid walls. These first-order asymptotics provide effective initial choices for a Newton–Raphson scheme, allowing the accurate numerical computation of the eigenfunctions within a few iterations. The orthogonality and completeness of the resulting eigenfunction basis are subsequently examined: orthogonality is shown to hold in a very general setting, beyond the specific case of admittance boundary conditions for the Helmholtz problem, while completeness is proven for rectangular rooms. The extension of this result to more general settings remains instead an open problem (see, for instance, Open Problem 4.7 in \cite{bogli2022eigenvalues}). An auxiliary eigenvalue problem is further investigated, yielding additional insights on the original problem.

Numerical experiments first verify the proposed asymptotic expressions for the eigenfunctions, confirming that all four asymptotic families may be required in practice. The Green's function evaluation algorithm is then benchmarked against standard numerical solvers, demonstrating its superior accuracy and computational efficiency. Auxiliary eigenvalues are also computed and shown to coincide with peaks in the transfer function, with their locations shifting toward lower or higher frequencies depending on the admittance values. Finally, the computed Green’s function is compared with experimental measurements conducted in a 3D rectangular room, showing good agreement.

In conclusion, owing to its robustness and fast convergence, the proposed algorithm represents a powerful alternative to conventional numerical solvers and can also serve as a reliable benchmark for numerical simulations. The present analysis further provides a foundation for future work, such as the development of second-order asymptotic approximations or the extension of the completeness theorem to more general geometries.

\section*{Acknowledgments}
The authors gratefully acknowledge HBK (Nærum, Denmark) for providing access to measurement equipment used in this work. The authors would also like to thank Roland Badeau and Finn T. Agerkvist for valuable discussions.

\section*{Author declarations}
\subsection*{Conflict of interest}
The authors declare no conflict to disclose.

\section*{Data availability}
Data will be made available on request.

\bibliographystyle{abbrv}
\bibliography{sampbib}

\end{document}